\documentclass[12pt]{article}
\usepackage{amsmath,amsthm}
\usepackage{url}
\usepackage{slanting}
\usepackage{graphicx}

\renewcommand{\(}{\left(}
\renewcommand{\)}{\right)}
\newcommand{\E}{{\mathbf{E}}}

\newcommand{\A}{\mathcal{A}}
\newcommand{\C}{\mathcal{C}}
\newcommand{\K}{\mathcal{K}}

\newcommand{\col}{\mathrm{colour}}

\newtheorem{theorem}{Theorem}
\newtheorem{prop}[theorem]{Property}
\newtheorem{prb}[theorem]{Problem}
\theoremstyle{definition}
\newtheorem*{dnt}{Definition}
\newtheorem{exm}[theorem]{Example}
\newtheorem{alg}[theorem]{Algorithm}

\author{%
Wojciech Kordecki\\
Department of Computer Science \\
Faculty of Technical and Economic Science \\
The Witelon State University of Applied Sciences in Legnica \\
e-mail: wojciech.kordecki@pwsz-legnica.eu
\and Anna {\L}yczkowska-Han{\'c}kowiak \\
Faculty of Informatics and Electronic Economy \\
Pozna{\'n} University of Economics \\
e-mail: anna.lyczkowska-hanckowiak@ae.poznan.pl}

\title{Greedy online colouring with buffering}

\begin{document}

\maketitle

\begin{abstract}
We consider the problem of online graph colouring.  Whenever a node is requested, a colour   must be assigned to the node, and this colour must be different from the colours of any of its neighbours. According to the greedy algorithm the node is coloured by the colour with the smallest possible $k$. 

The goal is to use as few colours as possible. We propose an algorithm, where the node is coloured not immediately, but only after the collection of next requests stored in the buffer of size~$j$. In other words, the first node in the buffer is coloured definitively taking into account all possible colourisations of the remaining nodes in the buffer. If there are $r$ possible corrected colourings, then the one with the probability $1/r$ is chosen. The first coloured  node is removed from the buffer to enable the entrance of the next request. A number of colours in a two examples of graphs: crown graphs and Kneser graphs have been analysed.

\end{abstract}

\section{Introduction}
\label{s:intro}

An online colouring of a graph $G$ is the one assigned to $G$ by colouring its vertices in some order 
\[
v_1,v_2\dots,v_n\,.
\]
The colour of $v_i$ is assigned by only looking at the subgraph
of $G$ induced by the set $\{v_1,v_2,\dots,v_i\}$, and the assigned colour of $v_i$ is never changed.  
Greedy colouring is a colouring of the vertices of a graph formed by a greedy algorithm that considers the vertices of the graph in sequence and assigns to each vertex its first available colour $k$. Another name used for the such an algorithm is \textit{First Fit} one.
Of course, greedy colourings do not generally use the minimum number of colours possible. 

The unpublished review paper by Miller \cite{Miller04onlinegraph} contains introductory information the comprehensive form and presents the main problems considered in this paper.
In \cite{Miller04onlinegraph}, Miller introduces the problem formally and he defines a performance metric to evaluate the success of an online colouring algorithm. 
He points out that online (greedy, first fit)  algorithm performs very well in the cases where the input graph belongs to a certain class of graphs.

Nevertheless, in many cases such an algorithm works very badly. The best known example of such a graph is a crown graph (see \cite{Miller04onlinegraph}). 
The algorithm with buffering presented in Section~\ref{s:buffering} essentially improves the effectiveness of colouring in the worst case, even for the buffer of a very small size.

Let $\A$ be an online algorithm used for colouring the graph $G$. Denote
\begin{itemize}
\item 
$\chi\(G\)$ -- chromatic number of $G$,
\item 
$\chi_{\A}\(G\)$ -- the maximum number of used colours for each possible ordering
of the vertices (the worst-case).
\end{itemize}
The performance ratio of an online graph colouring algorithm $\A$ for a class of graphs $\C$ is defined as
\begin{equation}\label{eq:perf_ratio}
\rho\(G\)=\max_{G\in\C}\left\{\frac{\chi_{\A}\(G\)}{\chi\(G\)}\right\}.
\end{equation}

Follow\cite{Miller04onlinegraph}~, we present two theorems by Halld{\'o}rsson and Szegedy: \cite{Halldorsson:OnLineColoringHyper},
\cite{HaSze:Lower} and
\cite{Halldorsson:OnlineColoringKnown}
\begin{theorem}
The performance ratio of any deterministic online colouring algorithm is at least 
\[
\frac{2n}{\ln^2n}\,.
\]
\end{theorem}

\begin{theorem}
The expected performance ratio of any randomised online colouring algorithm is at least 
\[
\frac{n}{16\ln^2n}\,.
\]
\end{theorem}

It is known that for any bipartite graph on $n$ vertices and any deterministic algorithm at least
\[
1.13747\cdot  \log_2 n-0.49887 
\]
colours are needed: Bianchi et al.~\cite{BBHK:bipartite}. 

Lov{\'a}sz, Saks, Trotter \cite{LST:on-line} prove
(see also Kierstead and Trotter \cite{KiersteatTrotter:On-line}, Bianchi~et~al.~\cite{BBHK:bipartite}):
\begin{theorem}\label{thm:LST:bipartite}
For any bipartite graph on $n$ vertices there exists an online algorithm using at most $2\log^*_2 n\(o\(1\)+1\)$ colours. 
\end{theorem}
Binary iterated logarithm $\log^*_2$ is the number of times the logarithm function must be iteratively applied before the result is less or equal to 1, i.e.
$\log^*_2 n=k$ where $k$ is the smallest number for which $k$ times iterated logarithm of $n$ is at most 1: 
\[
\underbrace{\log_2\dots\log_2 n}_{k\text{ times}}\leq 1.
\]
It is essential to know that arriving vertices are from the bipartite graph.

Colouring online has many real applications. See for example Bartal et al.~\cite{BartalFiatLeonardi:lowerbounds} and Zang et al. \cite{Zang_et_al}.
The recent articles, e.g. Bianchi et al.~\cite{BBHK:bipartite} and Christ et al.~\cite{Christ_et_al:Online} present an another look at the problem of optimising a number of colours using the so called ``bit advice''. In~\cite{Christ_et_al:Online} one can also find an interesting application of this method to the cellular networks.

The next step to the reality is allowing that the vertices can be not only coloured but they can also be discoloured (see Borowiecki and Sidorowicz~\cite{BorowieckiSidorowicz:dynamic}). 
Dynamic graph colourings can be naturally applied in system modeling, e.g. for scheduling threads of parallel programs, time sharing in wireless networks, session scheduling in high-speed LAN's, channel assignment in WDM optical networks as well as traffic scheduling.

In this paper we focus our attention on the case when we now know not only the present arriving vertex but we also know in advance the vertices which will arrive in the next several moments. In Section~\ref{s:buffering} we present an algorithm and analyse two particular known classes of graphs: crown graphs and Kneser graphs. In Section~\ref{s:simul} we present some numerical results obtained by computer simulations.

\section{Colouring with buffering}
\label{s:buffering}

The problem of online colouring with buffering is known as \textit{lookahead} and has been considered in the case of $d$-inductive graphs in Irani~\cite{Irani:inductive} (the review in Miller~\cite{Miller04onlinegraph}), 
also Halld{\'o}rsson~\cite{Halldorsson:OnLineColoringHyper} for hypergraphs.

The \textit{$d$-inductive graph} is a graph with a numbered sequence of vertices in such a way that every vertex is joined with the vertex with a maximal number by at most $d$ edges.
Irani showed that the greedy algorithm uses $O\(d\log n\)$ colours on $G$ if $G$ belongs
to the class of $d$-inductive graphs \cite{Irani:inductive}. Thus the performance ratio of the greedy algorithm on chordal and planar graphs is bounded above by $O\(\log n\)$.

%Let us formulate the algorithm of online colouring with buffering.
\begin{alg}(online colouring with buffering).
\label{alg:buffer}

\noindent
Let $B\subset V$ be the buffer of size $b$.
\begin{enumerate}
\item 
Fix a maximal size of the buffer $B$ as $b\geq 1$. 
\item 
Let $V_c=\emptyset$ be the set of already coloured vertices and set $B=\emptyset$.
\item 
Colour the first vertex by the colour 1 and move it to $V_c$.
\label{alg:fill}\item 
Fill the buffer by subsequent vertices as a queue FIFO until the buffer is full.
\item\label{alg:colour} 
If the buffer is full, colour the vertices in the queue properly (including $V_c$) using colours of the minimal values.
\item\label{alg:choose}
Among all possible colourings of the buffer choose only such ones  whose subsequent colours from the biggest one to the smallest one are minimal.
\item 
If such possible colourings are $r$, choose one with the probability $1/r$.
\item 
Colourings of all the vertices in the buffer except the first vertex is temporary. At the moment when the next vertex arrives, move the first one to $V_c$ and repeat the procedure of colouring.
\end{enumerate}
\end{alg}

\section{Analysis of special cases}
\label{s:special}

\subsection{Crown graph}
\label{ss:crown}

\begin{dnt}
A crown graph $\C_n=\(V,E\)$ on $2n$,  vertices is an undirected graph with two sets of vertices, $V=V_1\cup V_2$ with an edge from $v_{1,i}$ to $v_{2,j}$ whenever $i\neq j$. The crown graph can be viewed as a complete bipartite graph from which the edges of a perfect matching have been removed.
\[
V_k=\left\{v_{k,1},v_{k,2},\dots,v_{k,n}\right\}
\]
and
\[
\(u,w\)\in E\iff u=v_{1,i}, w=v_{2,j}, i\neq j.
\]
\end{dnt}

Crown graphs can be used to show that greedy colouring algorithms behave badly in the worst case: if the vertices of a crown graph are presented to the algorithm in the order $u_0$, $v_0$, $u_1$, $v_1$, etc., then a greedy colouring uses $n$ colours, whereas the optimal number of colours is two. This construction is attributed to Johnson \cite{Johnson:Worst}; crown graphs are sometimes called Johnson's graphs with notation $J_n$. F{\"u}rer \cite{Furer:Improved} uses crown graphs as part of a construction showing hardness of approximation of colouring problems.

Let us denote
\begin{align*}
V_1&=\left\{v_{1,1},v_{1,2},\dots,v_{1,n}\right\} \\
V_2&=\left\{v_{2,1},v_{2,2},\dots,v_{2,n}\right\}
\end{align*}

Let 
$\C_n$ have an linear order if
\[
V=\(v_{1,1},v_{1,2},\dots,v_{1,n},v_{2,1},v_{2,2},\dots,v_{2,n}\).
\]
and
$\C_n$ have an alternate order if 
\[
V=\(v_{1,1},v_{2,1},v_{1,2},v_{2,2},\dots,v_{1,n},v_{2,n}\).
\]

The following theorems show, how the size of the buffer affects the performance ratio.
\begin{theorem}\label{thm:CRn2_E}
If $\C_n$, $n\geq 2$ has the alternate order, then for Algorithm~\ref{alg:buffer} with  $b=2$ we have
\begin{equation}\label{eq:CRn2_E}
\E C_n=3-\frac{1}{2^n}\,.
\end{equation}
\end{theorem}
\begin{proof}
At the every level a number of used colours is increased by 1, if the right and left vertex have the same colour.
Such a situation can occur with the probability $1/2$ under condition that at every lower level the left and right vertices obtained the same colours.
If at the left and right vertices at the lower level have the different colours then the number of colours does not increase.
Therefore
\begin{equation}\label{eq:CRn2_p}
\Pr\(C_n=k\)=
\begin{cases}
\frac{1}{2^{k-1}} & \text{dla $1<k<n$,} \\
\frac{1}{2^{k-2}} & \text{dla $k=n$.}
\end{cases}
\end{equation}
Hence
\[
\E C=\sum_{k=2}^{n-1}\frac{k}{2^{k-1}}+\frac{n}{2^{n-2}}=3-\frac{1}{2^n}\,,
\]
which proved the formula~\eqref{eq:CRn2_E}.
\end{proof}

\begin{theorem}\label{thm:CRn2_P}
If $\C_n$, $n\geq 2$ has the alternate order, then for Algorithm~\ref{alg:buffer} with  $b=2$ and $m<n$ we have
\begin{equation}\label{eq:CRn2_P}
\Pr\(C_n\geq m\)=\frac{1}{2^{m-2}}\,.
\end{equation}
\end{theorem}
\begin{proof}
Formula~\eqref{eq:CRn2_P} follows immediately from \eqref{eq:CRn2_p} in the proof of Theorem~\ref{thm:CRn2_E}
\[
\Pr\(C_n\geq m\)=\sum_{k=m}^{n-1}\frac{1}{2^{k-1}}+\frac{1}{2^{n-2}}
\]
Since
\[
\sum_{k=m}^{n-1}\frac{1}{2^{k-1}}=\frac{1}{2^{m-2}}\(1-\frac{1}{2^{n-m}}\),
\]
\[
\sum_{k=m}^{n-1}\frac{1}{2^{k-1}}+\frac{1}{2^{n-2}}=\frac{1}{2^{m-2}}\,.
\]
\end{proof}

\begin{prop}\label{thm:CRn3}
Buffer of the size $b=3$ does not decrease the number of colours relatively to the buffer of the size $b=2$ for $\C_n$ with an alternate order.
\end{prop}
\begin{proof}
The fact that we know two next vertices in a crown graph with an alternate order does not give any additional information, because the last arrived vertex and the next third are always not joined.
\end{proof}

\begin{prop}\label{thm:CRn4}
If $\C_n$, $n\geq 2$, has the alternate order then for Algorithm~\ref{alg:buffer} with the buffer  $b=4$ the number of used colours is always equal to 2.
\end{prop}
\begin{proof}
It easy to observe that the four subsequent vertices in the crown graph with the alternate order give full information that the vertices in the buffer form a bipartite graph.
\end{proof}

Note that if $b=2$ it may occur that Algorithm~\ref{alg:buffer} give the worse colourisation than in the case $b=1$. 
\begin{exm}
In Figure~\ref{fig:cr4} the labels of vertices have the form \textsf{n:L}, where $n$ is the number of a subsequent arriving vertex and the letter $L$ denotes the colour used by Algorithm~\ref{alg:buffer}.
Vertices are coloured by colours $A,B,C,D$.

In Figure~\ref{fig:cr4} the difference in the colouring process with $b=2$ in comparison with the case $b=1$ is such that the second and the third vertex have to obtain different colours.  If the colour $B$ was chosen (with probability $1/2$) for the second vertex then the third vertex has to obtain the colour $A$. A a result the Algorithm~\ref{alg:buffer} has to give colours $C$ and $D$ for the last four vertices.

\begin{figure}[!hbt]
\centering
\includegraphics[width=4.5cm]{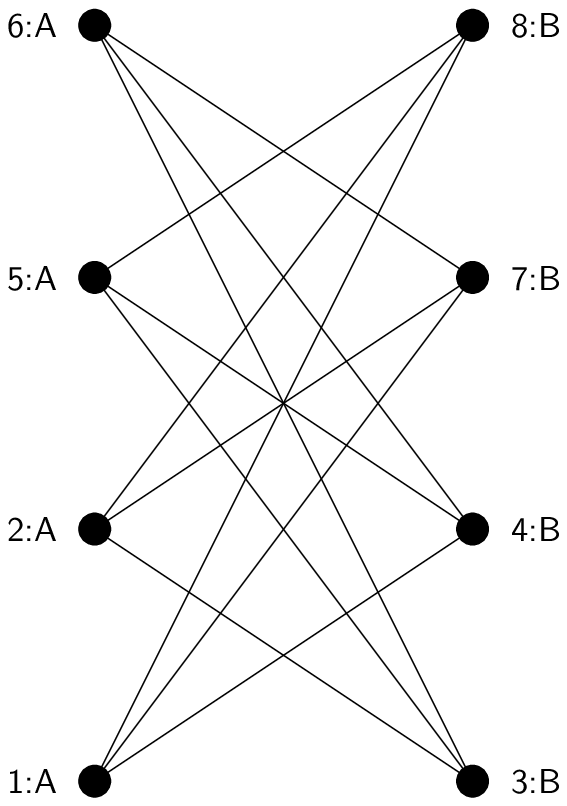}
\includegraphics[width=4.5cm]{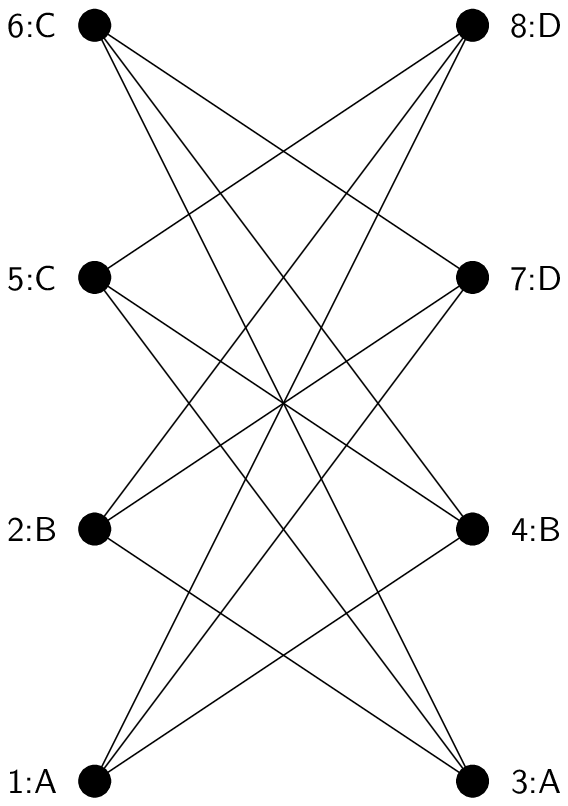}
\caption{\label{fig:cr4}Colouring online for the crown graphs $\C_4$ with buffering: $b=1$ (left) and $b=2$ (right)}
\end{figure}
\end{exm}

\subsection{Kneser graphs}
\label{ss:Kneser}

The vertices of $\K_{n,k}$ are all the $k$-element subsets of $\{1,2,\dots ,n\}$, and an
edge joins vertices $S$ and $T$ if and only if $S\cap T =\emptyset$. 
Such graphs were introduced by J.~Kneser in 1955 -- see~\cite{AZ:Proofs}, Section~38, p.~251.
Kneser conjectured that $\chi\(\K_{n,k}\)=n-2k+2$ for $n\geq 2$. This conjecture is proved by Lov{\'a}sz and with subsequent simpler proofs by B{\`a}r{\`a}ny and Matou{\v{s}}ek -- see~\cite{AZ:Proofs}. 
The class of Kneser graphs contains many familiar classes of graphs. 
\begin{itemize}
\item 
If $k > n/2$, then $\K_{n,k}$ is the empty graph. 
\item 
If $k=1$, then $\K_{n,k} = K_n$, the complete graph on $n$ vertices. 
\item 
$\K_{5,2}$ is the Petersen graph.
\end{itemize}
Miller in~\cite{Miller04onlinegraph} looks at the class of Kneser graphs as an interesting one and still unconsidered.

At first let us consider the simplest example, i.e. Petersen graph, using the greedy algorithm with the buffer of size $b=1$ and $b=2$.
Let us assume that vertices of Petersen graph are numbered in the lexicographic order:
\begin{equation}\label{eq:Petersen_num}
v_1=\{1,2\}, v_2=\{1,3\},\dots,v_{10}=\{4,5\}.
\end{equation}

In the following example we point out that a buffering can both decrease and increase the necessary number of colours. However, as we point out through simulations in Section~\ref{s:simul}, the average number colours used by our algorithm with the buffer of size $b=2$ is a bit smaller than the average number colours given with the buffer of size $b=1$, i.e. without buffering. Therefore we can formulate the following problem.

\begin{prb}
Determine the smallest $b$ that 
\begin{equation}\label{eq:}
\frac{\E C_{n,k}^{\(2\)}}{n-2k+2}-\frac{\E C_{n,k}^{\(b\)}}{n-2k+2}>\delta
\end{equation}
for some fixed $\delta$.
\end{prb}

\begin{exm}
In Figures~\ref{fig:petersen1} and \ref{fig:petersen2} the labels of vertices have the form \textsf{v-n:L}, where
$v$ is the number of vertex in Petersen graph as in Equation~\eqref{eq:Petersen_num}, $n$ is the number of a subsequent arriving vertex and a letter $L$ denotes the colour used by this algorithm.
Let us colour the vertices by colours $A,B,C,D$.

Assume that the order of arriving vertices is $8,1,5,7,6,2,10,4,3,9$. In Figure~\ref{fig:petersen1} the difference in the colouring process with $b=2$ in comparison with the case $b=1$ is such that if the present vertex is 6, we also know that the next vertex is 2. Since these vertices have to obtain different colours, we have two possibilities according to point~\ref{alg:colour} of Algorithm~\ref{alg:buffer}:
\begin{enumerate}
\item 
$\col\(5\)=A$, $\col\(2\)=C$ (as in the case $b=1$),
\item 
$\col\(5\)=B$, $\col\(2\)=A$.
\end{enumerate}
According to point~\ref{alg:choose} we choose the second possibility. Therefore, using the buffer of size 2, we can paint Petersen graph using three instead  four colours.

Using the buffer of the size two at least, we do not always obtain a better result.
Let as assume that the order of arriving vertices is $9,7,5,8,1,6,3,2,4,10$. In the case colouring with the buffer of size $b=1$ gives three colours. In the case $b=2$, if we colour the fourth vertex $v_8$ and we know that the next vertex is $v_1$, then we have to colour these vertices by $B$ and $C$. If we decide (with probability $1/2$) that $\col\(v_8\)=C$, then finally we must use four colours to paint Petersen graph instead of three colours. Such a case is presented on Figure~\ref{fig:petersen2}.

\begin{figure}[!hbt]
\centering
\includegraphics[width=6.5cm]{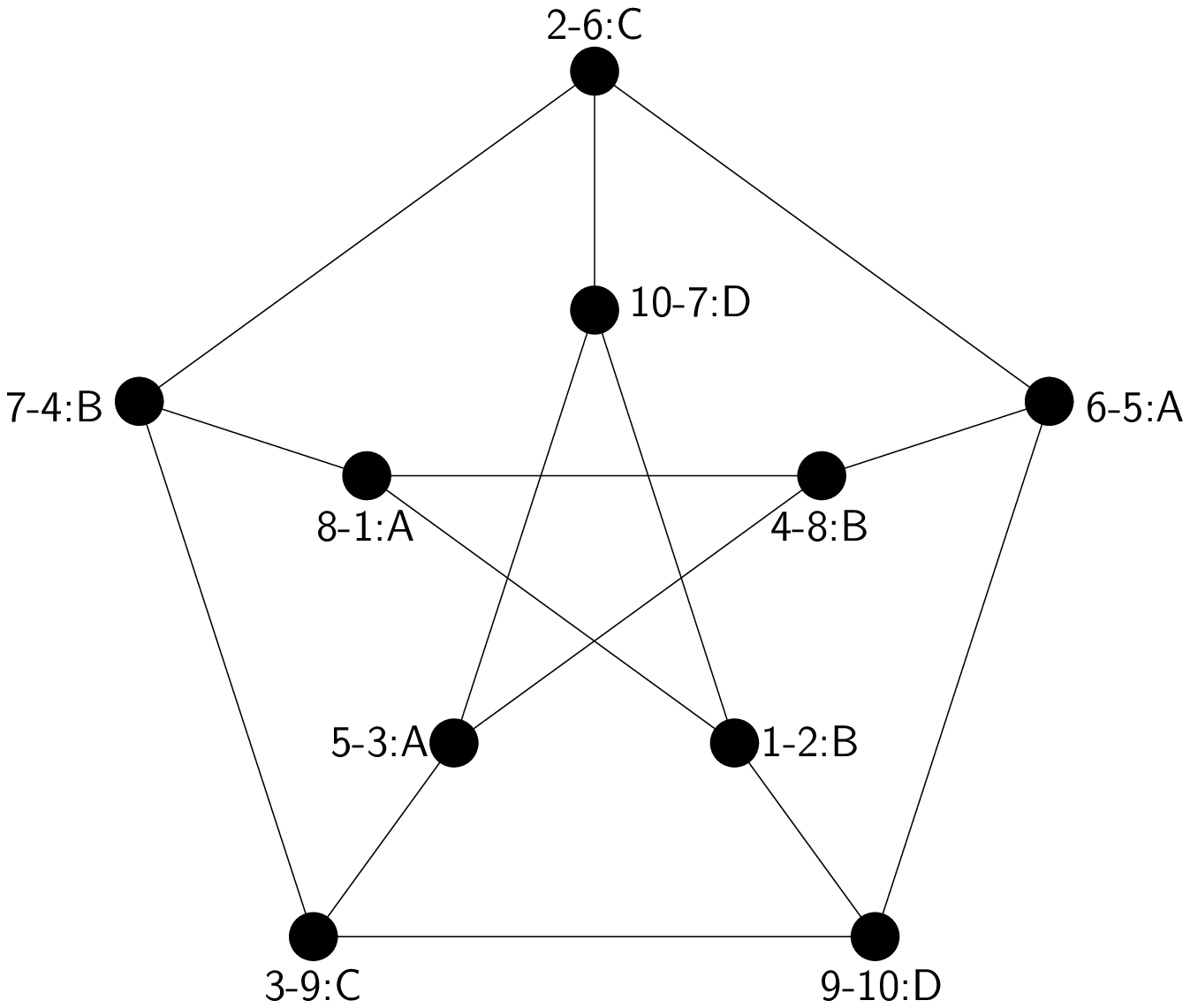}
\includegraphics[width=6.5cm]{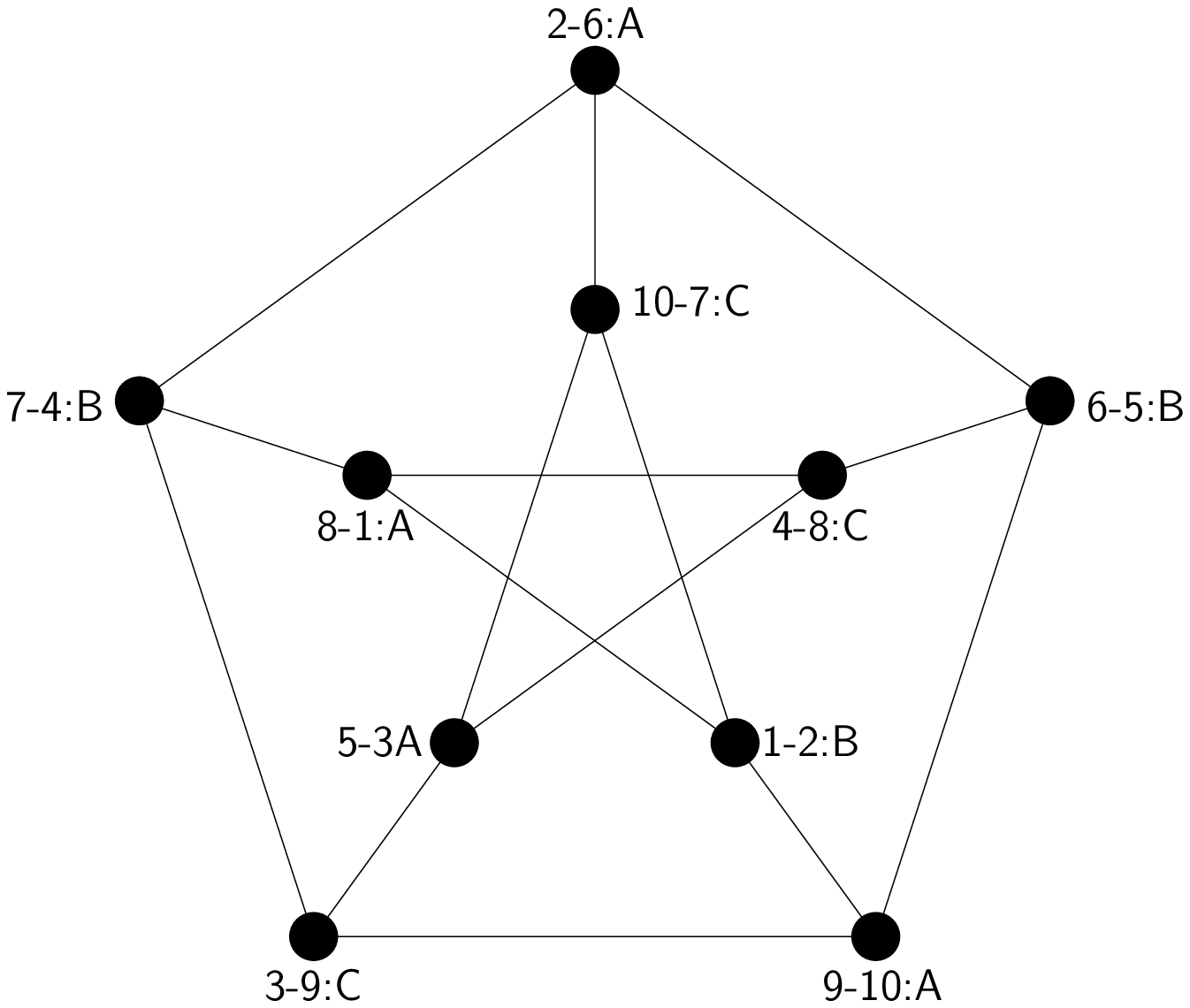}
\caption{\label{fig:petersen1}Colouring online for Petersen graphs with buffering: $b=1$ (left) and $b=2$ (right) and ordering $8,1,5,7,6,2,10,4,3,9$}
\end{figure}

\begin{figure}[!hbt]
\centering
\includegraphics[width=6.5cm]{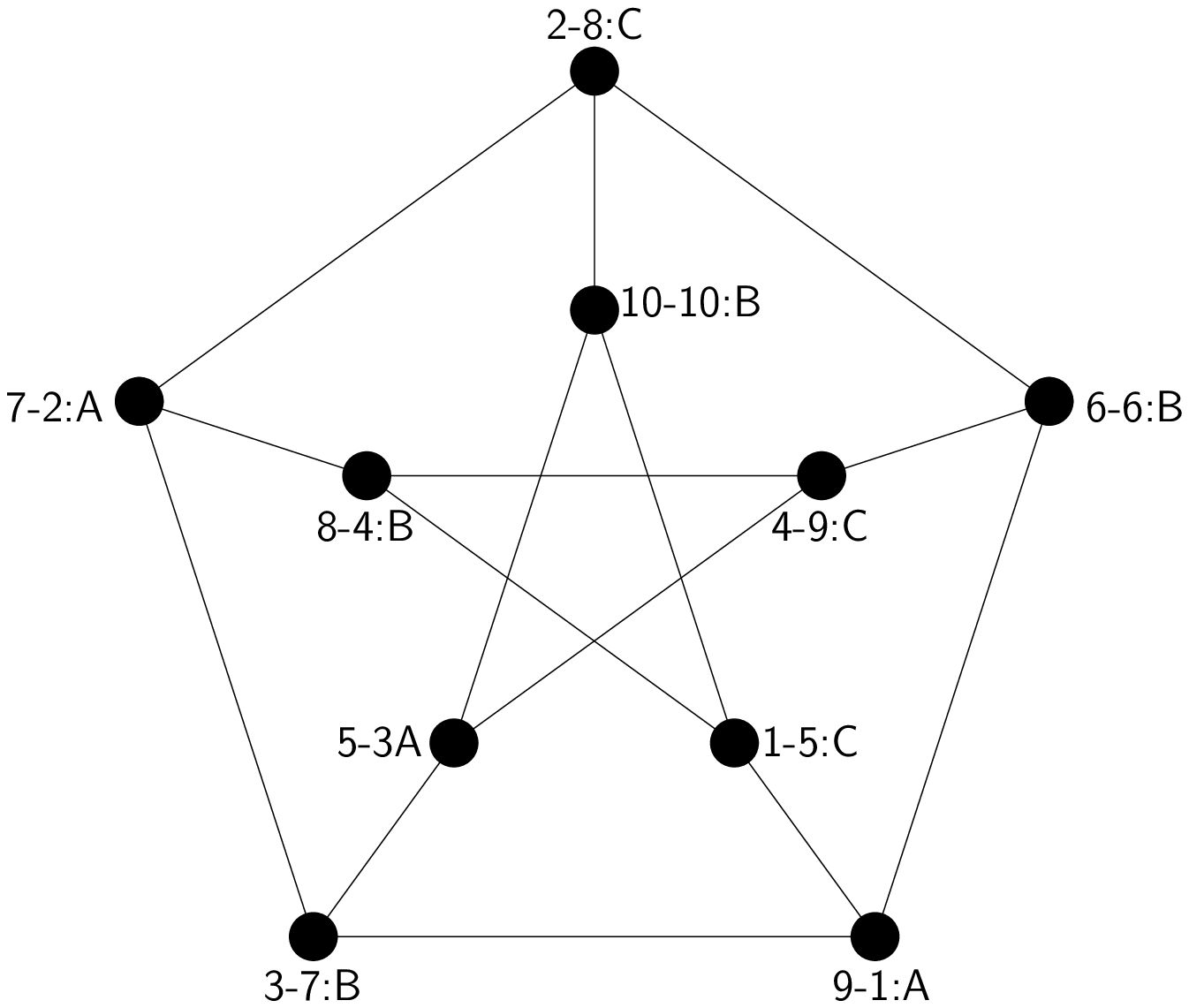}
\includegraphics[width=6.5cm]{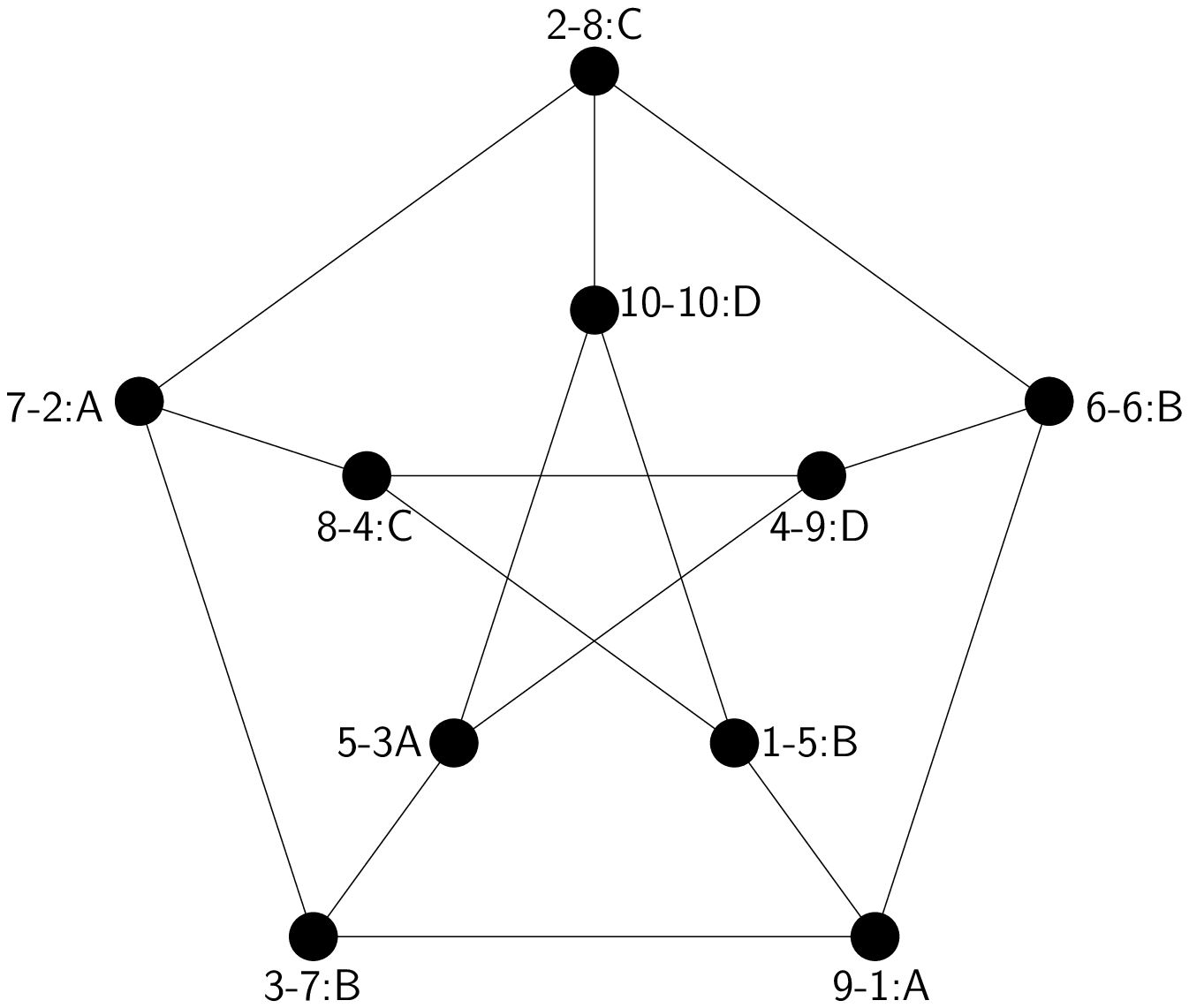}
\caption{\label{fig:petersen2}Colouring online for Petersen graphs with buffering: $b=1$ (left) and $b=2$ (right) and ordering $9,7,5,8,1,6,3,2,4,10$}
\end{figure}

\end{exm}

\section{Simulations}
\label{s:simul}

Simulations for crown graphs and Kneser graphs used $20\,000$ repetitions\footnote{The computer program was written in Pascal using Lazarus environment and the standard random number generator.}.
Taking into account of theorems~\ref{thm:CRn2_E}--\ref{thm:CRn4} we carried out the simulations for crown graphs only for $b\leq 2$. 
In the both considered cases the maximal number of vertices is equal to 200.

\begin{table}[!htb]
\begin{center}
\caption{\label{tab:simCrown}The average number of colours: result of simulations for random order of $\C_n$}

\smallskip
\begin{tabular}{r|r|r|r}
$n$ & $b=2$ & $b=1$ \\
\hline 
4 & 2.24 & 2.32 \\
6 & 2.15 & 2.20 \\
10 & 2.08 &  2.12 \\
20& 2.04 & 2.05 \\
50& 2.02 & 2.02 \\
100& 2.01 & 2.01 \\
\end{tabular}
\end{center}
\end{table}

For Kneser graphs, simulations were carried out for $5\leq n\leq 10$, $k\geq 2$ and $b\leq 2$.
The result of simulations for $b=1$ and $b=2$ is given in Table~\ref{tab:simKneser1}.
Simulations for $b=2$ give smaller but almost the same results as in the case when $b=1$.

\begin{table}[!hbt]
\caption{\label{tab:simKneser1}
The average number of colours: result of simulations for random order of $\K_{n,k}$}

\smallskip
\begin{center}
\begin{tabular}[t]{r|cc|cc|cc}
& 
\multicolumn{2}{c|}{$k=2$} &  
\multicolumn{2}{c|}{$k=3$} &
\multicolumn{2}{c}{$k=4$} \\\cline{1-7}
{\footnotesize \ukosna{c}{25}{15}{0}$b$\t1 $n$\t2}
& $1$ & $2$ & $1$ & $2$ & $1$ & $2$ 
\\\hline
 5  &3.13&3.10 && \\
 6  &4.28&4.23 && \\
 7  &5.44&5.37  &3.93&3.918 \\
 8  &6.58&6.51  &5.69&5.65 \\
 9  &7.70&7.64  &7.35&7.30  &4.89&4.88 \\
10  &8.81&8.74  &8.93&8.88  &7.52&7.49 \\
\end{tabular} 
\end{center}
\end{table}

%\newpage
\bibliographystyle{abbrv}
\bibliography{CID}

\end{document}